\documentclass[11pt,a4paper,reqno]{amsproc}

\makeatletter

\usepackage{datetime}
\usepackage[pdfborder={0 0 0}]{hyperref}
  \hypersetup{
    citebordercolor=.1 .6 1,
    colorlinks = false
}

\usepackage{upref,cases}
\hypersetup{citebordercolor=.1 .6 1}

\usepackage[T1]{fontenc}
\usepackage[utf8]{inputenc}

\usepackage[final]{graphicx}

\IfFileExists{mmymtpro2.sty}{\usepackage[subscriptcorrection]{mymtpro2}

  \def\cup{\cupprod}
  \def\cap{\capprod}
  \def\bigcup{\bigcupprod}
  
  \def\bigcupdisjoint{\mathop{\kern10pt\raisebox{4pt}{$\cdot$}\kern-12pt\bigcup}\limits}
}%
	{\usepackage{amssymb,bm,dsfont,fourier,times}

	\let\wtilde\widetilde
}

\let\epsilon\varepsilon

\DeclareMathOperator{\Tr}{Tr}

\newcommand{\abs}[1]{\ensuremath{\left\vert #1 \right\vert}}

\newcommand{\ceiling}[1]{\ensuremath{\left\lceil #1 \right\rceil}}

\DeclareMathOperator{\supp}{supp}

\numberwithin{equation}{section}
\newtheoremstyle{ttheorem}%
       {1.8ex\@plus1ex}                
       {2.1ex\@plus1ex\@minus.5ex}      
       {\itshape}           
       {0pt}                   
       {\bfseries}          
       {.}                  
       {.5em}               
       {}                

\newtheoremstyle{ddefinition}%
       {1.8ex\@plus1ex}                
       {2.1ex\@plus1ex\@minus.5ex}      
       {}           
       {0pt}                   
       {\bfseries}           
       {.}                  
       {.5em}               
       {}                

\newtheoremstyle{rremark}%
       {1.8ex\@plus1ex}                
       {2.1ex\@plus1ex\@minus.5ex}      
       {\normalfont}        
       {0pt}                   
       {\bfseries}           
       {.}                  
       {.5em}               
       {}                   

\theoremstyle{ttheorem}
\newtheorem{theorem}{Theorem}[section]
\newtheorem{lemma}[theorem]{Lemma}
\newtheorem{proposition}[theorem]{Proposition}
\newtheorem{corollary}[theorem]{Corollary}

\theoremstyle{ddefinition}
\newtheorem{definition}[theorem]{Definition}

\theoremstyle{rremark}

\newtheorem{myremarks}[theorem]{Remarks}
\newtheorem{myexamples}[theorem]{Examples}
\newtheorem{example}[theorem]{Example}

\newenvironment{remarks}{\begin{myremarks}\begin{nummer}}%
    {\end{nummer}\end{myremarks}}
    {\end{nummer}\end{myexamples}}

\newcounter{numcount}
\newcommand{\labelnummer}{(\roman{numcount})}%

\makeatletter
\providecommand{\showkeyslabelformat}[1]{\relax}        
\let\mysaveformat\showkeyslabelformat                   %
\def\myformat#1{\raisebox{-1.5ex}{\mysaveformat{#1}}}   %

\newenvironment{nummer}%
  {\let\curlabelspeicher\@currentlabel%
    \begin{list}{\textup{\labelnummer}}%
      {\usecounter{numcount}\leftmargin0pt%
        \topsep0.5ex\partopsep2ex\parsep0pt\itemsep0ex\@plus1\p@%
        \labelwidth2.5em\itemindent3.5em\labelsep1em%
      }%
    \let\saveitem\item%
    \def\item{\saveitem%
      \def\@currentlabel{\curlabelspeicher\kern.1em\labelnummer}}%
    \let\savelabel\label%
    \def\label##1{{\ifnum\thenumcount=1\let\showkeyslabelformat\myformat\fi\savelabel{##1}}%
										{\def\@currentlabel{\labelnummer}%
									 	\let\showkeyslabelformat\@gobble
									 	\savelabel{##1item}%
										}%
	   							}%
  }{\end{list}}%

  {\let\curlabelspeicher\@currentlabel%
    \begin{list}{\textup{\labelnummer}}%
      {\usecounter{numcount}\leftmargin0pt%
        \topsep0.5ex\partopsep2ex\parsep0pt\itemsep0ex\@plus1\p@%
        \labelwidth2.5em\itemindent0em\labelsep1em%
        \leftmargin2.5em}%
    \let\saveitem\item%
    \def\item{\saveitem%
      \def\@currentlabel{\curlabelspeicher\kern.1em\labelnummer}}%
    \let\savelabel\label%
    \def\label##1{{\ifnum\thenumcount=1\let\showkeyslabelformat\myformat\fi\savelabel{##1}}%
										{\def\@currentlabel{\labelnummer}%
									 	\let\showkeyslabelformat\@gobble
									 	\savelabel{##1item}%
										}%
    							}%
  }{\end{list}}%

\usepackage{enumitem}

\let\OldItem\item
\newcommand{\MyItem}[2][]{}%
\def\section{\@startsection{section}{1}%
  \z@{1.3\linespacing\@plus\linespacing}{.5\linespacing}%
  {\normalfont\bfseries\centering}}

\def\subsection{\@startsection{subsection}{2}%
  \z@{.8\linespacing\@plus.5\linespacing}{-1em}%
  {\normalfont\bfseries}}

\def\nlsubsection{\@startsection{subsection}{2}%
  \z@{.8\linespacing\@plus.5\linespacing}{.1ex}%
  {\normalfont\bfseries}}

\let\@afterindenttrue\@afterindentfalse%

\renewenvironment{proof}[1][\proofname]{\par \normalfont
  \topsep6\p@\@plus6\p@ \trivlist 
  \item[\hskip\labelsep\scshape
    #1\@addpunct{.}]\ignorespaces
}{%
  \qed\endtrivlist
}

\def\ps@firstpage{\ps@plain
  \def\@oddfoot{\normalfont\scriptsize \hfil\thepage\hfil
     \global\topskip\normaltopskip}%
  \let\@evenfoot\@oddfoot
  \def\@oddhead{
    \begin{minipage}{\textwidth}
      \normalfont\scriptsize
      \emph{\insertfirsthead}
    \end{minipage}}
  \let\@evenhead\@oddhead 
}

\def\insertfirsthead{}

\def\@cite#1#2{{%
 \m@th\upshape\mdseries[{#1}{\if@tempswa, #2\fi}]}}
\newcommand{\C}{\mathbb{C}}
\newcommand{\N}{\mathbb{N}}
\newcommand{\PP}{\mathbb{P}}
\newcommand{\R}{\mathbb{R}}
\newcommand{\Z}{\mathbb{Z}}
\renewcommand{\L}{\Lambda}

\renewcommand{\leq}{\leqslant}
\renewcommand{\geq}{\geqslant}
\providecommand{\wtilde}[1]{\widetilde{#1}}

\let\<\langle
\let\>\rangle

\providecommand{\abs}[1]{\lvert#1\rvert}

\newcommand{\1}{1}

\newcommand{\upd}{\mathrm{d}}
\renewcommand{\d}{\upd}   

\newcommand{\hairspace}{\kern .04167em}

\newcommand{\beq}{\begin{equation}}
\newcommand{\eeq}{\end{equation}}

\newcommand{\E}{\mathbb E}

\def\clap#1{\hbox to 0pt{\hss#1\hss}}

\makeatother
\usepackage{todonotes}

\begin{document}


\title[Lifshitz tails for random perturbations of Laurent matrices]{Lifshitz tails for random diagonal perturbations of Laurent matrices}

\author[M.\ Gebert]{Martin Gebert}
\address[M.\ Gebert]{Mathematisches Institut,
  Ludwig-Maximilians-Universit\"at M\"unchen,
  Theresienstra\ss{e} 39,
  80333 M\"unchen, Germany}
\email{gebert@math.lmu.de}

\author[C.\ Rojas-Molina]{Constanza Rojas-Molina}
\address[C.\ Rojas-Molina]{Laboratoire AGM, Dept. de Math\'ematiques, CY Cergy Paris Universit\'e, 2 Av. Adolphe Chauvin, 95302 Cergy-Pontoise, France}
\email{crojasmo@cyu.fr}

\thanks{Version \today}

\maketitle

\begin{abstract} We study the Integrated Density of States of one-dimensional random operators acting on $\ell^2(\mathbb Z)$ of the form $T + V_\omega$ where $T$ is a Laurent (also called bi-infinite Toeplitz) matrix and $V_\omega$ is an Anderson potential generated by i.i.d. random variables. We assume that the operator $T$ is associated to a bounded, H\"older-continuous symbol $f$, that attains its minimum at a finite number of points. We allow for $f$ to attain its minima algebraically. The resulting operator $T$ is long-range with weak (algebraic) off-diagonal decay. We prove that this operator exhibits Lifshitz tails at the lower edge of the spectrum with an exponent given by the Integrated Density of States of $T$ at the lower spectral edge. The proof relies on generalizations of Dirichlet-Neumann bracketing to the long-range setting and a generalization of Temple's inequality to degenerate ground state energies.
\end{abstract}


\section{Introduction}

Recently, fractional Anderson operators of the form $H_{\omega,\alpha}=(-\Delta)^{\alpha/2}+V_\omega$, with $\alpha\in(0,2]$ and $V_\omega$ an Anderson potential generated by i.i.d. random variables, have been studied in both the discrete setting $\ell^2(\mathbb Z^d)$ and continuous setting $\rm{L}^2(\mathbb R^d)$, with $d\geq 1$. They have attracted interest for their connections to anomalous diffusion \cite{MR3787555,MR3916700,MR3600570,Riascos_2018}, the appearance of an asymptotic behavior for the Integrated Density of States known as fractional Lifshitz tails \cite{kaleta2019lifschitz,kaleta2019lifschitzb,GRM20} and for exhibiting weaker localization properties compared to the standard Anderson model in $d=1$ \cite{Padgett_Liaw_etal19}.

In \cite{GRM20} we studied the Integrated Density of states (IDS) $N$ of the fractional Anderson model $H_{\omega,\alpha}$ in the discrete setting, showing fractional Lifshitz tails at the bottom of the spectrum. Compared to the behavior of the IDS for the Anderson model $H_{\omega,2}= -\Delta + V_\omega$ at the bottom of the spectrum
\beq\label{eq:ltand} N(E)\sim e^{-E^{-\frac{d}{2}}},\quad E\searrow 0,
\eeq
called Lifshitz tails, the IDS of the fractional Anderson model $H_{\omega,\alpha}$, $\alpha\in(0,2)$ exhibits the behavior
\beq N(E)\sim e^{-E^{-\frac{d}{\alpha}}},\quad E\searrow 0,
\eeq
called \emph{fractional} Lifshitz tails where $\sim$ denotes double logarithmic asymptotics, see \eqref{double_log_asymp}. The factors $d/2$ and $d/\alpha$, respectively, are called Lifshitz exponents.
In \cite{GRM20}, our proof was based on the operator monotonicity of the map $x\mapsto x^s$, $s\in[0,1]$ and the standard Dirichlet-Neumann bracketing available for finite-volume restrictions of the discrete Laplacian. This result is in agreement with results obtained recently in the continuous setting using probabilistic arguments that are not directly applicable in the discrete setting \cite{kaleta2019lifschitz}.

In this note we consider a more general framework than in \cite{GRM20} by studying random diagonal perturbations of Laurent matrices (also called bi-infinite Toeplitz matrices) associated to a symbol $f\in L^\infty([-\pi,\pi])$ satisfying suitable regularity conditions. In our setting we consider operators of the form $H_\omega=T_f+V_\omega$ where $T_f$ is a long-range operator and $V_\omega$ is an Anderson potential generated by i.i.d. random variables. The off-diagonal decay of $T_f$ can be weak, as long as it is summable. The symbol $f$ associated to $T_f$ can have a finite number of minima, as long as the minima are attained algebraically. We show that in this case, the IDS $N$ of the operator $H_\omega$  exhibits Lifshitz tails behavior at the bottom of the spectrum with a Lifshitz exponent given by the behavior of the IDS of the free operator $T_f$,
\beq\label{eq:lt-laur} N(E)\sim e^{-\frac{1}{I_f(E)}},\quad E\searrow \inf f, \eeq
where $I_f$ is the IDS of the Laurent operator $T_f$ given by
\beq
I_f(E) = \frac{1}{2\pi}\big| \big\{ k\in[-\pi,\pi]: f(k) \leq E\big\}\big|
\eeq
and $|\cdot|$ denotes Lebesgue measure.

 Our proof is based on a novel Dirichlet-Neumann bracketing obtained recently in \cite{G20} for suitably defined finite-volume restriction of banded Toepliz matrices, on a generalization of Temple's inequality to degenerate ground state energies, and on the operator monotonicity of the map $x\mapsto x^s$, $s\in[0,1]$, generalizing the approach of \cite{GRM20}. In the case of Toeplitz matrices with exponentially decaying off-diagonal terms perturbed by an Anderson potential, it was shown in \cite{Kl98} using periodic approximations of the IDS, that the latter exhibits exhibits Lifshitz tails of the form \eqref{eq:lt-laur}. Our novel approach allows us to recover the results from \cite{Kl98} and treat more general cases, as the case where $T_f$ is the fractional Laplacian, and more general functions of the discrete negative Laplacian.

The article is organized as follows: in the next section we introduce the model and state the main result on Lifshitz tails for a random perturbation of a Laurent matrix. In Section \ref{s:upplow-f} we show upper and lower bounds on the symbol $f$ associated to the Toeplitz matrix which we are interested in. In Section \ref{s:dn-bracketing} we recall some results from \cite{G20} on the Dirichlet-Neumann bracketing for banded Toeplitz matrix and show that we can bound the IDS of our model above and below by the IDS of an auxiliary model consisting of banded Toeplitz matrices perturbed by an Anderson potential. In Sections \ref{s:upp} and \ref{s:low} we give the proofs of upper and lower bounds on the IDS of the auxiliary model and finally in Section \ref{s:proof} we bring all together to give a proof of Lifshitz tails for our model.

\section{Model and main result}\label{s:model}
We start this section by introducing the free operator that will later be perturbed by a random potential.

Let $\mathbb T=\mathbb R \diagup 2\pi\mathbb Z$ be the one-dimensional torus and consider a function $f\in \rm{L}^\infty(\mathbb T)$. We are interested in the Laurent (also called bi-infinite Toeplitz) operator $T_f$ associated to $f$ on $\ell^2(\Z)$. $T_f$ is defined, for $x\in \ell^2(\Z)$, by
\beq
(T_f x)_n = \sum_{m\in\Z} a_{m-n} x_m,
\eeq
where the sequence $\big(a_n\big)_{n\in\Z}$ consists of the Fourier coefficients of the function $f$, that is,
\beq
a_n=\frac 1 {2\pi} \int_{-\pi}^\pi \d k f(k) e^{-i k n}.
\eeq
The function $f$ is called the symbol of the operator $T_f$.

Note, moreover, that $T_f$ is unitarily equivalent to the operator $M_f$ given by multiplication by $f$ on $L^2(\mathbb T)$. More precisely, $T_f$ is diagonalized by the discrete Fourier transform $\mathcal F:\ell^2(\mathbb Z)\rightarrow\rm{L}^2(\mathbb T)$ given by $(\mathcal{F}u)(k)=\frac{1}{\sqrt{2\pi}} \sum_{n\in \mathbb Z} u(n)e^{-ink}$, where $u\in \ell^2(\mathbb Z)$ and $k\in \mathbb T$. The fact that $f\in {\rm L}^\infty(\mathbb Z)$ implies that $T_f$ is a bounded operator, as $M_f$ is a bounded operator and its spectrum, denoted by $\sigma(T_f)$, is given by the range of $f$, that is, $\sigma(T_f)=f(\mathbb T)$.

\noindent {\bf Assumptions:} 
\textit{We assume throughout that the symbol $f$ satisfies
\begin{itemize}
\item[(A1)] $f$ is real valued.
\item[(A2)] $f$ is $\nu$-H\"older-continuous for some $\nu>0$ and $f\in C_{\text{pw}}^1(\mathbb T)$, where $f\in C_{\text{pw}}^1(\mathbb T)$ means $f$ being piecewise continuously differentiable, i.e. continuously differentiable except at finitely many points. Note that $f$ is $\nu$-H\"older continuous at the points where it fails to be continuously differentiable.
\item[(A3)] $\displaystyle\min_{x\in \mathbb T} f(x) = 0$ and there exists $M\in\mathbb N$ such that the minimum of $f$ is attained at $M$ points $E_1,...,E_M$.
\item[(A4)] The minima are attained algebraically, i.e., for each $E_i$, $i=1,...,M$ given in (A3), there exists $\beta_i\geq \nu$, $i=1,...,M$ such that the following limit exists and is positive
\beq
\lim_{\substack{E\to E_i,\\ E\neq E_i}} \frac{f(E)}{|E - E_i|^{\beta_i}} >0.
\eeq
We define $b:=\displaystyle\max_{1\leq i\leq M} \beta_i$.
\end{itemize}
}
 Assumption (A1) implies that $T_f$ is a self-adjoint operator, i.e. $a_n =\overline{ a}_{-n}$ for all $n\in\Z$.
 Therefore, $T_f$ is a bounded self-adjoint operator.

We denote by $(\delta_n)_{n\in\mathbb Z}$ the canonical orthonormal base of $\ell^2(\mathbb Z)$.
Assumption (A2) on the regularity of the symbol $f$ implies that the matrix entries of $T_f$ satisfy (see Lemma \ref{lem:offdiag-decay})
\beq\label{eq:offdiag-decay1} \langle \delta_n,T_f\delta_m \rangle=\abs{a_{m-n}}\leq \frac{1}{\abs{m-n}^{1+\nu}} \quad,m, n \in\mathbb Z.\eeq
\begin{remarks}
\item
The case $f(k) = 2- 2 \cos(k)$ gives rise to the discrete negative one-dimensional Laplacian, i.e. $T_f=-\Delta$. The function $f(k)= (2- 2 \cos(k))^\alpha$ with $0<\alpha<1$ gives rise to the discrete fractional negative Laplacian, i.e., $T_f=(-\Delta)^\alpha$.
\item We can consider symbols of the form $f=\Phi((2- 2 \cos(k)))$, where $\Phi:[0,\infty)\rightarrow [0,\infty)$ is a complete Bernstein function satisfying $c_1\lambda^{\alpha/2}\leq \Phi(\lambda)\leq c_2 \lambda^{\alpha/2}$  for all $\lambda<\lambda_0$ and some $\alpha\in(0,2]$ and $c_1,c_2,\lambda_0>0$,
     as considered in \cite{kaleta2019lifschitzb} (see also \cite{Riascos_2018}). This gives rise to operators of the form $T_f=\Phi(-\Delta)$.
\item
The prime example satisfying (A1)\,--\,(A4) is  the symbol
\beq
f(x) = \prod_{i=1}^M \big(2-2\cos(x-E_i)\big)^{\alpha_i}
\eeq
for some distinct $E_1,...,E_M\in\mathbb T$ and some  $\alpha_1,...,\alpha_M >0$.
\item
Condition (A4) excludes that the function $f$ approaches zero as $e^{-\frac 1 x}$.
\end{remarks}

Next, we define a diagonal random perturbation of the operator $T_f$.
\begin{definition}
Given a symbol $f$ satisfying Assumptions (A1) -- (A4), we define the random Laurent operator $H_{f,\omega}$, acting on $\ell^2(\mathbb Z)$, by
\beq
H_{f} = T_f + V_\omega,
\eeq
where $T_f$ is a Laurent operator generated by the symbol $f$ and $V_\omega$ is an Anderson random potential of the form
\beq
V_\omega u(n)=\sum_{n\in\mathbb Z}\omega_n \vert \delta_n\rangle\langle\delta_n\vert
\eeq
with  $\omega:=(\omega_n)_{n\in\mathbb Z}\in{\mathbb R}^{\mathbb Z}$ being independent and identically distributed random variables distributed according to the Borel probability measure $\mathbb P=\bigotimes_{\mathbb Z} P_0$ on $\mathbb R^{\mathbb Z}$. The single-site probability measure $P_0$ is non trivial and we assume the infimum of $\supp P_0$ is zero. We denote the corresponding expectation by $\mathbb E(\cdot)$.
\end{definition}

The fact that $T_f$ is translation invariant and the assumptions on the random potential imply that  $H_{f}$ is an ergodic, bounded, self-adjoint operator and, by standard arguments, its spectrum $\sigma(H_{f})$ is deterministic (see e.g. \cite{MR1223779,MR2509110}). This, together with Assumption (A3) on the symbol $f$, and $\inf \supp P_0 = 0$ yields that  $\inf\sigma(H_{f})=\inf\sigma(T_f)=0$.


\subsection{The integrated density of states}\label{s:IDS}

Let $L\in\N$ and write $\L_L=[-L,L]\cap\mathbb Z$. We denote the restriction $T_{f,L}$ of $T_f$ to $\ell^2\big(\Lambda_L\big)$ with simple boundary conditions by $T_{f,L}= 1_{\L_L} T_f 1_{\L_L}$ where $1_S$ stands for the projection onto $\ell^2(S)\subset \ell^2(\Z)$ for $S\subset\Z$. We denote the eigenvalues of $T_{f,L}$ by $\lambda_1\leq ...\leq \lambda_{2L+1}$, $j=1,..,2L+1$ counting multiplicity and in non-decreasing order.

The Integrated Density of States (IDS) $I_f$ of $T_f$ is defined by
\beq\label{eq:IDS-Tf}
I_f(E):=\lim_{L\to\infty} \frac{\#\{j:\ \lambda_j\leq E\} }{2L+1}.
\eeq
The translation invariance of $T_f$ and the off-diagonal decay \eqref{eq:offdiag-decay1} implies that
this limit exists and
\beq
I_f(E)= \frac 1 {2\pi}\big| \big\{ k\in [-\pi,\pi]: \ f(k) \leq E\big\}\big|,
\eeq
where $|\cdot|$ denotes the Lebesgue measure of a set, see Proposition \ref{prop:IDS-Tf}. Assumption (A4) on the symbol $f$ implies that
\beq
I_f(E) = \frac 1 {2\pi} \big| \big\{ k\in [-\pi,\pi]: \ f(k) \leq E\big\}\big| \sim C  E^{\frac 1 b} + o(E^{\frac 1 b})
\eeq
as $E\searrow 0$ with $b = \displaystyle\max_{1\leq i \leq M} \beta_i$.

Introducing randomness, we are interested in the behaviour of the IDS at the lower edge of the spectrum.
As before, we consider the restriction $H_{f,L}$ of $H_f$ to $\ell^2\big(\L_L\big)$ with simple boundary conditions, given by $H_{f,L}= 1_{\L_L} H_f 1_{\L_L}$. We denote the eigenvalues of $H_{f,L}$ by
\beq
\mu_1\leq \mu_2\leq ... \leq \mu_{2L+1}.
\eeq
Now, the IDS  $N_f$ of $H_f$ takes the form
\beq
 N_f(E):=\lim_{L\to\infty} \frac{\#\{j:\ \mu_j\leq E\} }{2L+1}=\lim_{L\to\infty} \frac{ \Tr \big( 1_{\leq E}  (H_{f,L})\big)}{2L+1}
\eeq
To see that this limit exists, we note the following equivalent representation of the function $N_f$:
 \begin{lemma}\label{lem:IDS}
 Almost surely, the limit
\beq \label{eq:ids-alt}
\lim_{L\to\infty} \frac{\Tr\big( 1_{[-L,L]} 1_{\leq E} (H_{f}) \big)}{2L+1}
\eeq
exists and equals $N_f$ for all $E\in\mathbb R$.
\end{lemma}
\begin{proof}[Proof of Lemma \ref{lem:IDS}]
The limit exists and is independent of the particular random realization due to the ergodicity of $H_f$ (see, e.g. \cite{MR1223779,MR2509110}).

To show that
\eqref{eq:ids-alt} equals $N_f$, we use \eqref{eq:offdiag-decay1}, proven in Lemma \ref{lem:offdiag-decay} and follow the same arguments as in \cite[Proposition 2.1]{GRM20}.
\end{proof}

Our main result is a statement on the behavior of the IDS of $H_f$ near its spectral infimum $\inf \sigma(H_f)=0$.

\begin{theorem}\label{thm:main}
Under the assumptions (A1) -- (A4) on the symbol $f$, the IDS of $H_f$ satisfies at the lower edge of the spectrum
\beq
\limsup_{E\searrow 0} \frac{\ln | \ln N_f(E) |}{\ln E} \leq -\frac 1 b,
\eeq
where $b = \displaystyle \max_{1\leq i \leq M} \beta_i$.

If, moreover, the single-site probability distribution satisfies $P_0([0,\epsilon))\geq C \epsilon^\kappa$ for some $C,\kappa>0$, we have the equality
\beq\label{double_log_asymp}
\lim_{E\searrow 0} \frac{\ln | \ln N_f(E) |}{\ln E} = -\frac 1 b.
\eeq
\end{theorem}

The theorem above shows that, under the stated conditions on $f$, as $E\searrow 0$
\beq\label{eq:LTstrong}
N_f(E) \sim e^{-\frac 1 {I_f(E)}}
\eeq
in a weak double logarithmic way.
Therefore, the behavior of the IDS of $H_f$ at the bottom of the spectrum is determined by the behavior of the free operator $T_f$ there.

\section{Upper and lower bounds on the symbol $f$}\label{s:upplow-f}

\begin{lemma}\label{lm:upp_low}
Let $f\in L^\infty([-\pi, \pi])$ satisfying Assumptions $(A1)-(A4)$. Then there exist constants $0<c\leq C$ such that for all $x\in[-\pi,\pi)$
\beq
c\prod_{i=1}^n \big(2-2\cos(x-E_i)\big)^{b/2}
\leq
f(x)
\leq
C\big(2-2\cos(x - E_{i_0})\big)^{b/2}
\eeq
where, as before, $b= \displaystyle\max_{1\leq i \leq M} \beta_i$ and $i_0 \in \big\{i: \beta_i = b\big\}$ where $\beta_1,..,\beta_M$ are the exponents given in assumption (A4).
\end{lemma}

\begin{proof}
We define for $x\in[-\pi,\pi]\setminus\{E_1,...E_M\}$
\beq
g(x) := \frac{f(x)}{\prod_{i=1}^M \big(2-2\cos(x-E_i)\big)^{\beta_i/2}}>0.
\eeq
where $E_i$, $\beta_i$, $i=1,...,M$ are as in Assumptions (A3)\,--\,(A4). Moreover, for $i=1,...,M$ we compute
\beq
\lim_{\substack{x\to E_i\\ x\neq E_i}} g(x) =
\frac 1 {\prod_{j=1,j\neq i}^M \big(2-2\cos(x-E_j)\big)^{\beta_j/2}}
\lim_{\substack{x\to E_i\\ x\neq E_i}} \frac{f(x)}{|x - E_i|^{\beta_i}} \frac{|x - E_i|^{\beta_i}}{\big(2-2\cos(x-E_i)\big)^{\beta_i/2}} .
\eeq
By Assumption $(A4)$ this limit exists and is strictly positive. Hence, we can extend $g$ to a continuous function on $[-\pi,\pi]$ with $g(x)>0$ for all $x\in[-\pi,\pi]$. Since $[-\pi,\pi]$ is compact, there exist
constants $c_1,C_1>0$ such that $c_1\leq g(x) \leq C_1$ for all $x\in [-\pi,\pi]$, that is,
\beq
c_1\prod_{i=1}^M \big(2-2\cos(x-E_i)\big)^{\beta_i/2}
\leq f(x)
\leq
C_1 \prod_{i=1}^M \big(2-2\cos(x-E_i)\big)^{\beta_i/2}.
\eeq
Since for $i=1,...,M$ and $x\in[-\pi,\pi]$  we have $0\leq 2-2\cos(x-E_i)\leq 4$, we further estimate from below
\beq
\prod_{i=1}^M \big(2-2\cos(x-E_i)\big)^{\beta_i/2} \geq 4^{\frac 1 2(\sum_i \beta_i -M b)} \prod_{i=1}^M \big(2-2\cos(x-E_i)\big)^{b/2}
\eeq
where $b= \max_i \beta_i$ and
\beq
\prod_{i=1}^M \big(2-2\cos(x-E_i)\big)^{\beta_i/2} \leq 4^{\sum_{i\neq i_0} \beta_i}\big( 2 - 2 \cos(x - E_{i_0})\big)^{b/2}
\eeq
where $i_0 \in \big\{i:\ \beta_i = b\big\}$. Setting $c := c_1 4^{\frac 1 2(\sum_i \beta_i -M b)}$ and $C:= C_1 4^{\sum_{i\neq i_0} \beta_i}$, gives the assertion.
\end{proof}

If the symbols $f_1,f,f_2\in\rm{L}^\infty(\mathbb T)$ satisfy $f_1\leq f \leq f_2$ , then the Laurent operators associated to them fulfil $T_{f_1}\leq T_f\leq T_{f_2}$ in operator sense. In turn, we have that $H_{f_1}\leq H_f \leq H_{f_2}$ in operator sense.  Lemma \ref{lm:upp_low} then yields the following

\begin{corollary}\label{corupplow}
Let $f_1(x) := c\prod_{i=1}^M \big(2-2\cos(x-E_i)\big)^{b/2}$,
$f_2(x) := C\big(2-2\cos(x - E_{i_0})\big)^{b/2}$ for $x\in[-\pi,\pi]$, and $0<c\leq C$ as in Lemma \ref{lm:upp_low}. Then for all $E\in\R$
\beq
N_{f_2}(E) \leq N_f(E) \leq N_{f_1}(E).
\eeq
\end{corollary}

\begin{example}
We illustrate the upper and lower bound on a given symbol with an example. Let $f:\mathbb T\to\R$ be the symbol
\beq
f(t) = 0.5 \big(2-2 \cos(t)\big)^{0.3} \big(2-2 \cos(t - 2.5)\big)^{0.6} \big(2-2 \cos(t + 2)\big)^{0.7}.
\eeq
Then we bound $f$ from below and above by
\begin{align}
0.5 \big(2-2 \cos(t)\big)^{0.7} &\big(2-2 \cos(t - 2.5)\big)^{0.7} \big(2-2 \cos(t + 2)\big)^{0.7} \notag\\
&
\leq f(t)
\leq\notag\\
3 &\big(2-2 \cos(t + 2)\big)^{0.7}.
\end{align}
We didn't optimize on the constants $0.5$ and $3$.
\end{example}
\begin{figure}[h]
\caption{Example of a symbol and its upper and lower bound }
\centering
\includegraphics[scale=.4]{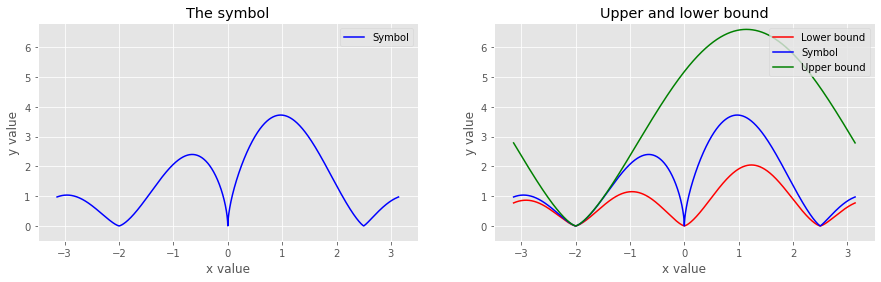}
\end{figure}

\section{Dirichlet-Neumann bracketing for banded Laurent matrices}\label{s:dn-bracketing}

We define for distinct $E_1,...,E_M$, $\alpha>0$ and $\mathfrak{c}>0$, the function $f_{\alpha,E_1,...,E_M}:[-\pi,\pi]\to \R$,
\beq\label{eq:def-f}
 f_{\alpha,E_1,...,E_M}(x) := \mathfrak{c}\prod_{i=1}^M \big(2-2\cos(x-E_i)\big)^{\alpha},
\eeq
and write
\beq\label{eq:f-g}
 f_{\alpha,E_1,...,E_M}(x) = \mathfrak{c}\Big(\prod_{i=1}^M \big(2-2\cos(x-E_i)\big)^{\overline \alpha } \Big)^\beta =: \mathfrak{c} g(x)^\beta
\eeq
where
$\overline \alpha:=\min\{n\in\N: \alpha\leq n\}\in\N$, $\beta:= \alpha/\overline\alpha\leq 1$ and
\beq\label{eq:def-g}
g(x):=\prod_{i=1}^M \big(2-2\cos(x-E_i)\big)^{\overline \alpha }.
\eeq
Note that $g$ involves integer powers of $2-2\cos(x-E_i)$, and its associated operator $T_g$  is a banded Laurent matrix with band width $2N+1$, where $N=M\overline{\alpha}$ (see \cite[Section 1]{G20}).

To study the IDS of functions of the type \eqref{eq:def-f}, it is useful to consider finite-volume restrictions with certain boundary conditions. We have encountered already the so-called simple boundary conditions, where $T_{g,[a,b]}$ stands for the restriction of $T_g$ to $[a,b]\cap\mathbb Z$ given by $T_{g,[a,b]}=1_{[a,b]}T_g1_{[a,b]}$. In what follows, we will introduce other boundary conditions that will allow us to obtain lower and upper bounds on the IDS of functions of type \eqref{eq:def-f}, obtained in \cite{G20}.

We recall the notion of modified Dirichlet (respectively Neumann) boundary conditions for Laurent matrices given in \cite[Section 2]{G20}. For a banded Laurent matrix with band size $2N+1$, imposing boundary conditions consists in adding Hermitian $N\times N$ matrices at the endpoints of the respective boundary. For $a<b$, $a,b\in\mathbb Z\cup{\pm \infty}$ we define the restriction of $T_g$ to the interval $[a,b]\cap\mathbb Z$ with boundary condition $\star$ at the left/right endpoint as
 \beq\label{def:bc}
	 T^{\star,0}_{g,[a,b]} :=
	T_{g,[a,b]} +
	\begin{pmatrix}
		B_\star  & 0 \\
		0 & 0
	\end{pmatrix}
	\quad\text{and}\quad
	T^{0,\star}_{f,[a,b]} :=
	T_{g,[a,b]} +
	\begin{pmatrix}
		0  & 0 \\
		0 & \wtilde B_\star
	\end{pmatrix}
\eeq
where, $B_\star$ is an $N\times N$ matrix and $\wtilde{ B}_\star$ is the reflection of $B_\star$ along the anti-diagonal, i.e. $\wtilde{ B}_\star := U^* B_\star U$ with $U: \C^N\to\C^N$, $ (Ux)_k := x_{N-k+1}$ for $x = (x_1,...,x_N)\in \C^N $. If the boundary conditions are imposed on both left and right endpoints, we write $T^{\star,\star}_{g,[a,b]}$ which is of the form
\beq T^{\star,\star}_{g,L} :=
	T_{g,[a,b]} +
	\begin{pmatrix}
		B_\star  & 0 & 0\\
		0 & 0 & 0 \\
		0 & 0 & \wtilde B_\star
	\end{pmatrix}.
\eeq
In the latter, the $0$s stand for matrices consisting of zeros in the respective size.
In the case where either $a=-\infty$ or $b=+\infty$, we write $T^{\star}_{g,(-\infty,b]}$ to indicate that the boundary condition is defined at the point $b$, respectively $T^{\star}_{g,[a,\infty)}$ for when it is defined at the point $a$.

By \cite[Remarks 2.3]{G20}, the modified Dirichlet boundary conditions $\mathcal D$ are given by non-negative definite matrices $B_{\mathcal D}, \wtilde B_{\mathcal D}$, and the Neumann boundary conditions $\mathcal N$ are given by  $B_{\mathcal N} = -B_{\mathcal D}, \wtilde B_{\mathcal N} = - \wtilde B_{\mathcal D}$ and are therefore non-positive. The matrices are defined in such a way that they satisfy the following Dirichlet--Neumann bracketing:

\begin{theorem}[Theorem 1.1 \cite{G20}]
Let $M\in\mathbb N$, $E_1,...,E_M\in\mathbb T$ distinct and $\alpha>0$. Let $g$ be the function given in \eqref{eq:def-g} and $N=M\overline\alpha$. There exists boundary conditions $\mathcal N$ and $\mathcal D$ such that
\beq\label{eq:d-n-bracketing-T}
T_{g,\L_{L_1}}^{\mathcal N,\mathcal N}\oplus T_{g,\Lambda_{L_2}}^{\mathcal N,\mathcal N} \leq T_{g,\Lambda_{L}} \leq T_{g,\Lambda_{L_1}}^{\mathcal D,\mathcal D} \oplus T_{g,\Lambda_{L_2}}^{\mathcal D,\mathcal D}
\eeq
for all $L_1,L_2\in\Z$ with $L_1,L_2\geq 2N+1$ such that $\L_{L_1},\L_{L_2}$ are disjoint intervals with $\L_{L_1}\cup \L_{L_2}=\L_L$.
\end{theorem}
From the positive-definiteness (respectively, negative-definiteness) of the boundary conditions $\mathcal D$ (respectively $\mathcal N$) we readily conclude that for $L\in\mathbb N$, $\L_L=[-L,L]\cap\mathbb Z$ and $\L^c_L=\mathbb Z\setminus \L_L=(-\infty,-L-1)\cup(L+1,\infty)$
\beq
T_{g,\L_L}^{\mathcal N,\mathcal N}\oplus T_{g,\L_L^c}^{\mathcal N}\leq T_g \leq T_{g,\L_L}^{\mathcal D,\mathcal D}\oplus T_{g,\L^c_L}^{\mathcal D},
\eeq
where $T_{g,\L^c_L}^{\star}$ denotes the restriction of $T_g$ to $\L_L^c$ with boundary conditions $\star\in\{\mathcal D,\mathcal N\}$ at the endpoints $-L-1$ and $L+1$.
\begin{corollary}\label{cor:upp-low}
Let $M\in\mathbb N$, $E_1,...,E_M\in\mathbb T$ distinct, $\alpha>0$ and $\mathfrak{c}>0$. Let $f_{\alpha,E_1,...,E_M}$ and $g$ be the functions given in \eqref{eq:def-f} and \eqref{eq:def-g}, respectively, and $N=M\overline\alpha$. For $0<\beta\leq 1$ we have for all $L\geq 2N+1$ that
\beq
\mathfrak{c} \left(T_{g,\L_L}^{\mathcal N,\mathcal N}\right)^\beta \leq 1_{\L_L}T_{f_{\alpha,E_1,...,E_M}}1_{\L_L} \leq  \mathfrak{c} \left(T_{g,\L_L}^{\mathcal D,\mathcal D}\right)^\beta.
\eeq
\end{corollary}
\begin{proof}
Using the fact that the map $[0,\infty)\to\R:x\mapsto x^\beta$ with $\beta\in(0,1]$ is operator monotone \cite[Thm. V.2.10]{Bha97}, and that the operators involved in the direct sum act as block diagonal matrices, we obtain
\beq
\left(T_{g,\L_L}^{\mathcal N,\mathcal N}\right)^\beta\oplus \left(T_{g,\L_L^c}^{\mathcal N,\mathcal N}\right)^\beta\leq \left(T_g\right)^\beta \leq \left(T_{g,\L_L}^{\mathcal D,\mathcal D}\right)^\beta\oplus \left(T_{g,\L_L^c}^{\mathcal D,\mathcal D}\right)^\beta.
\eeq
Projecting on both sides of the inequalities onto $\L_L$, we obtain
\beq
\left(T_{g,\L_L}^{\mathcal N,\mathcal N}\right)^\beta \leq 1_{\L_L}\left(T_g\right)^\beta1_{\L_L} \leq  \left(T_{g,\L_L}^{\mathcal D,\mathcal D}\right)^\beta.
\eeq
Using the fact that the discrete Fourier transform is a unitary operator, one can see that $\mathfrak{c}\left(T_{g}\right)^\beta=\mathfrak{c}T_{g^\beta}= T_{\mathfrak{c}g^\beta} =T_{f_{\alpha,E_1,...,E_M}}$, which yields the desired result.
\end{proof}
We obtain the following lower and upper bounds for the IDS of a function of type \eqref{eq:def-f}.
\begin{corollary}\label{upper_lower_bound}
For distinct $E_1,...,E_M$, $\alpha>0$, consider the functions $f_{\alpha,E_1,...,E_M}$ and $g$ given in \eqref{eq:def-f} and \eqref{eq:def-g} and let $N = M\overline \alpha $. Then, for all $L\geq 2N+1$, $\Lambda=[-L,L]\cap\mathbb Z$ we have
\beq\label{eq:DNbracketing-IDS}
\frac 1{|\Lambda|}\mathbb E\big(\Tr \1_{(-\infty,E]}\big(\mathfrak{c}(T^{\mathcal D,\mathcal D}_{g,\Lambda})^\beta + V_\omega \big)
\big)\leq
N_{ f_{\alpha,E_1,...,E_n}}(E)
\leq
\frac 1{|\Lambda|}\mathbb E\big(\Tr \1_{(-\infty,E]}\big(\mathfrak{c}(T^{\mathcal N,\mathcal N}_{g,\Lambda})^\beta + V_\omega \big)\big)
\eeq
\end{corollary}
\begin{proof}
   Note that for the modified boundary conditions, if we write the cube $\L_L$ as a disjoint union of smaller cubes $\L_j$, $\L_L=\cup_j \L_j$, we have $T^{\mathcal D,\mathcal D}_{L_L}\leq \bigoplus_j T^{\mathcal D,\mathcal D}_{L_j}$ and $T^{\mathcal N,\mathcal N}_{L_L}\geq \bigoplus_j T^{\mathcal N,\mathcal N}_{L_{j}}$ for $\L_L=\cup_j \L_j$. Therefore, for $\beta\in(0,1)$,
   \[ \left(T^{\mathcal D,\mathcal D}_{L_L}\right)^\beta\leq \bigoplus_j \left(T^{\mathcal D,\mathcal D}_{L_j}\right)^\beta  \quad \mbox{and}\quad \left(T^{\mathcal N,\mathcal N}_{L_L}\right)^\beta\geq \bigoplus_j \left(T^{\mathcal N,\mathcal N}_{L_{j}} \right)^\beta.  \]
   This, combined with Corollary \ref{cor:upp-low} gives the claim.
\end{proof}

\section{Upper bound on the IDS}\label{s:upp}
Let us recall that by Corollary \ref{corupplow} the IDS of the operator $T_f$ associated to a symbol $f$ satisfying Assumptions (A1)\,--\,(A4) in Section \ref{s:model} is bounded above by
\beq N_f(E)\leq N_{f_1}(E).
\eeq
$N_{f_1}$ is IDS of the operator $T_{f_1}$ associated to the symbol $f_1$ given by
\beq
f_1(x)=c\prod_{i=1}^M \big(2-2\cos(x-E_i)\big)^{b/2}
\eeq
where $b= \max_{1\leq i \leq M} \beta_i$, given in Assumption (A3) and $c>0$. Therefore $f_1$ is a function of type \eqref{eq:def-f} with $\alpha=b/2$. By Corollary \ref{upper_lower_bound}, the IDS $N_{f_1}$ is bounded above in \eqref{eq:DNbracketing-IDS} by the limit of the normalized eigenvalue counting function of the operator  $\big((T_{g,L})^{\mathcal N,\mathcal N}\big)^\beta$
\beq\label{eq:ub}
\frac 1{|\Lambda|}\mathbb E\big(\Tr \1_{(-\infty,E]}\big(c(T^{\mathcal N,\mathcal N}_{g,\Lambda})^\beta + V_\omega \big)\big)
\eeq
 where the symbol $g$ is related to $f_1$ by \eqref{eq:f-g}. We now proceed to obtain an upper bound on \eqref{eq:ub}
following the arguments in \cite[Section 6]{MR2509110}. The proof is given at the end of this section.

\subsection{Generalization of Temple's inequality}
One obstacle to continue with the standard proof of Lifshitz tails which uses Temple's inequality is that the degeneracy of the lowest eigenvalue of $T^{\mathcal N,\mathcal N}_{g,\Lambda}$. 

We first recall the following result from \cite{G20}, that gives a lower bound on the spectral gap above energy $0$ for the restriction of $T_g$ to a finite volume  $\Lambda_L=[-L,L]$ with modified Neumann boundary conditions defined in Section \ref{s:dn-bracketing}.

\begin{proposition}[Proposition 1.4 \cite{G20}]\label{prop:g20}
Let $M\in\mathbb N$, $E_1,...,E_M\in\mathbb T$ be distinct and $\alpha>0$. Let $g$ be of the form \eqref{eq:def-g} and let $N=M\overline\alpha$. We denote by $\lambda^{(L)}_1\leq ...\leq \lambda^{(L)}_{2L+1}$ the eigenvalues of $T_{g,L}^{\mathcal N,\mathcal N}$ ordered increasingly and counting multiplicity. Then $\lambda^{(L)}_k=0$ for $k=1,...,N$ and there exists $\tilde C>0$ such that for all $L\geq 2N+1$
\beq
\lambda^{(L)}_{N+1} \geq \frac {\tilde C}{L^{2\overline \alpha}}.
\eeq
\end{proposition}

If we denote by $\mu^{(L)}_1\leq ...\leq \mu^{(L)}_{2L+1}$ the eigenvalues of $c\big((T_{g,L})^{\mathcal N,\mathcal N}\big)^\beta$ ordered increasingly and counting multiplicity, then, by the Proposition above, $\mu^{(L)}_k=0$ for $k=1,...,N$ and there exists $C_0>0$ such that the spectral gap above the ground state energy satisfies
\beq\label{eq:gap-g}
\mu^{(L)}_{N+1} \geq \frac {C_0}{L^{2\alpha}}=\frac {C_0}{L^{b}}
\eeq
for all $L\geq 2N+1$.

For simplicity, we write $T_{L}^{\mathcal N} =  c\big((T_{g,L})^{\mathcal N,\mathcal N}\big)^\beta$ and denote by $E_0(T_L^{\mathcal N}+ \wtilde V_\omega\big)$ the ground state of $T_L^{\mathcal N}+ \wtilde V_\omega$.
We define
\beq\label{eq:wv}
\wtilde V_\omega := \min\big\{ \tilde c\Big(\frac{C_0}{L^{b}}\Big),V_\omega\big\} \leq V_\omega
\eeq
where $0<\tilde c<1$ is a constant to be determined later and, as before, $b= \max_{1\leq i \leq M} \beta_i$, given in Assumption (A4).

We have the bound
\beq\label{eq:bds-ub}
\frac 1{|\Lambda_L|}\mathbb E\big(\Tr \1_{(-\infty,E]}\big(T_L^{\mathcal D} + V_\omega \big)\big) \leq \mathbb P \big(E_0(T_L^{\mathcal N}+ V_\omega)<E\big)  \leq \mathbb P \big(E_0(T_L^{\mathcal N}+ \wtilde V_\omega)<E\big).
\eeq

\begin{theorem}[Generalization of Temple's inequality]\label{gen:temple}
Let $\wtilde V_\omega$ and $\tilde c$ be as in \eqref{eq:wv}, with $0<\tilde c < \frac 1 2$ and $L\in\N$ large enough. Then we obtain the lower bound
\beq
E_0(T_L^{\mathcal N}+ \wtilde V_\omega\big) + 6 \tilde c^2 \Big(\frac{C_0}{L^{b}}\Big)
\geq
(1-2 \tilde c)^2
\min_{\varphi\in \mathcal G, \|\varphi\|=1} \<\varphi,\wtilde V_\omega\varphi\>,
\eeq
where
\beq
\mathcal G:=\big\{\varphi\in \ell^2(\Lambda_L): T_L^{\mathcal N}\varphi = 0\big\}
\eeq
 is the ground-state space of the operator $T_L^{\mathcal N}$.
\end{theorem}

\begin{proof}
Let $t\in [0,1]$ and denote by $\mu_1(t) \leq ...\leq \mu_{2L+1}(t)$ the eigenvalues of $T_L^{\mathcal N}+ t\wtilde V_\omega$ ordered increasingly and counting multiplicity. The eigenvalues $(\mu_k(t))_{k=1,...,2L+1}$ are piecewise differentiable \cite[Chap. XII]{RS78}, and we write, using the Feynman-Hellmann theorem, see e.g., \cite{IZ88} ,
\beq
\mu_k(1) = \int_0^1 \<\varphi_k(t),\wtilde V_\omega \varphi_k(t)\>,
\eeq
where $\varphi_k(t)$ is the respective normalized eigenvector of $\mu_k(t)$, chosen to be continuous in $t$. Let $P$ be the projection onto $\mathcal G$ and $P^\perp = 1-P$. For the modified Neumann boundary conditions we have $0 = \mu_1(0)= ... =\mu_{N}(0)$, see Proposition \ref{prop:g20}. We write $\varphi(t) = \varphi_1(t)$ for $k=1$ and $t\in [0,1]$ and obtain the lower bound
\begin{align}
\<\varphi(t),\wtilde V_\omega \varphi(t)\>
= &
\<(P+P^\perp)\varphi(t),\wtilde V_\omega(P+P^\perp) \varphi(t)\> \notag\\
\geq &
\<P\varphi(t),\wtilde V_\omega P\varphi(t)\>  - |\<P^\perp\varphi(t),\wtilde V_\omega P\varphi(t)\>| \notag\\
& - |\<P\varphi(t),\wtilde V_\omega P^\perp\varphi(t)\>| - |\<P^\perp\varphi(t),\wtilde V_\omega P^\perp\varphi(t)\>|.
\end{align}
We estimate
\begin{align}
|\<P^\perp\varphi(t),\wtilde V_\omega P\varphi(t)\>|&+|\<P\varphi(t),\wtilde V_\omega P^\perp\varphi(t)\>|+|\<P^\perp\varphi(t),\wtilde V_\omega P^\perp\varphi(t)\>|\notag\\
&\leq
 3\|P^\perp P(t)\| \|\wtilde V_\omega\| 
 \label{eq:2}
\end{align}
where $P(t)$ stands for the projection onto $\mathcal G(t)=\text{span}\big\{\varphi_k(t): \varphi_k(0) \in\mathcal G\big\}$. The spectral gap between the first $N= \dim \mathcal G$ eigenvalues of $T_L^{\mathcal N}+ t\wtilde V_\omega$ and the rest of the spectrum is at $t=0$ bigger or equal to $C_0/L^{b}$, by \eqref{eq:gap-g}. Therefore for any $t\in[0,1]$
\beq
\mu_{N+1} (T_L^{\mathcal N}+ t\wtilde V_\omega) -\mu_{N} (T_L^{\mathcal N}+ t\wtilde V_\omega) \geq \frac{C_0}{L^{b}} -  t\|\wtilde V_\omega\|\geq \frac{C_0}{L^{b}} -  \|\wtilde V_\omega\|.
\eeq
Using the subspace perturbation bound in \cite[Chapter VII, Sec. 4]{Bha97}, we obtain for $\tilde c\leq \frac 1 2 $
\beq
\|P^\perp P(t)\| \leq \frac{\|\wtilde V_\omega\|}{\frac{C_0}{L^{b}} -  \|\wtilde V_\omega\|}\leq \frac{\tilde c}{1-  \tilde c }\leq 2 \tilde c,
\eeq
where we used $\|\wtilde V_\omega\|\leq \tilde c \frac{C_0}{L^b}$. Inserting this bound in \eqref{eq:2}, we obtain for all $t\in[0,1]$
\begin{align}
\eqref{eq:2}
\leq
6 \tilde c \|\wtilde V_\omega\|
\end{align}
and therefore,
\beq
\<\varphi(t),\wtilde V_\omega \varphi(t)\>
\geq \<P \varphi(t) ,  \wtilde V_\omega P\varphi(t)\>
- 6 \tilde c \|\wtilde V_\omega\| .
\eeq
We further estimate
\beq
\<P \varphi(t),\wtilde V_\omega P  \varphi(t)\>
\geq
\|P \varphi(t)\|^2 \min_{\varphi\in \mathcal G , \|\varphi\|=1} \<\varphi,\wtilde V_\omega\varphi\> .
\eeq
As before we estimate
\begin{align}
\|P \varphi(t)\| = \|PP(t)  \varphi(t)\|
= \|(1-P^\perp)P(t) \varphi(t)\|
&\geq \|P(t) \varphi(t)\| - \|P^\perp P(t) \|\notag\\
&\geq  1 - 2\tilde c.
\end{align}
All together we have proved that
\begin{align}
E_0(T_L^{\mathcal N}+ \wtilde V_\omega\big)
= \mu_1(1)
&\geq
(1-2 \tilde c)^2
\min_{\varphi\in \mathcal G , \|\varphi\|=1} \<\varphi,\wtilde V_\omega\varphi\>
- 6 \tilde c \|\wtilde V_\omega\|  \notag\\
&
\geq
(1-2 \tilde c)^2
\min_{\varphi\in \mathcal G , \|\varphi\|=1} \<\varphi,\wtilde V_\omega\varphi\>
- 6 \tilde c^2\Big(\frac{C_0}{L^{b}}\Big)
\end{align}
where we again used $\|\wtilde V_\omega\|\leq \tilde c \frac{C_0}{L^{b}}$ for the last term.
This implies the assertion.
\end{proof}

In order to further estimate  $\displaystyle\min_{\varphi\in \mathcal G, \|\varphi\|=1} \<\varphi,\wtilde V_\omega\varphi\>$ from below, we first recall from \cite[Section 5]{G20} that the ground-state space $\mathcal G$ is given by
\beq
\mathcal G = \text{span} \big\{\varphi^j_{k,L}\in \ell^2(\Lambda_L):\ k=1,...,M,\,\,j=0,...,\overline \alpha-1  \big\}
\eeq
with
\begin{align}\label{def:varphi_k1}
\varphi^j_{k,L} := \frac 1 {K_{j,L}^{1/2}}( (-L)^j e^{i L  E_k },...,0^j,1^j e^{-i E_k},..., L^je^{-i L E_k})^T\in\R^{2L+1}
\end{align}
 for $k=1,...,M$ and $j=0,...,\overline \alpha-1$, and normalization
 \beq
 K_{j,L} = 2\sum_{m=-L}^L |m|^{2j}.
 \eeq
 In particular $K_{0,L} = |\Lambda_L| = 2L+1$ and $K_{j,L} = O(L^{2j+1})$ as $L\to\infty$. Note that there are $M\overline \alpha =N$ vectors, therefore the dimension of $\mathcal G=N$.

Note that in \cite[Section 5]{G20} results are stated for a one-sided Toeplitz matrix, while in our case we work with a two-sided Toeplitz matrix, therefore the notation in \eqref{def:varphi_k1} differs slightly from the one in \cite[Eq. 5.1]{G20}.

\begin{lemma}\label{lower_bound_lm2}
For $a\in(0,1)$ depending only on $N$, there exists $C_{a,N}\in(0,1)$ depending only on $a,N$, and $L_0\in\N$ such that for all $L\geq L_0$ and all $\varphi\in\mathcal G$
\beq\label{lm_lower_bound0}
S^a_L:=\#\big\{ l\in\Lambda_L :\  |\Lambda_L||\varphi(l)|^2 \geq a  \big\} \geq C_{a,N} |\Lambda_L|
\eeq
and
\beq\label{lm_lower_bound}
\min_{\varphi\in \mathcal G, \|\varphi\|=1} \<\varphi,\wtilde V_\omega\varphi\> \geq
\frac a {|\Lambda_L|} \sum_{k\in S^a_L} \wtilde V _\omega(k).
\eeq
\end{lemma}
\begin{proof}
For fixed $a\in(0,1)$ possibly depending on $N$, we consider the sets
\[S^a_L=\#\big\{ l\in\Lambda_L :\  |\Lambda_L||\varphi(l)|^2 \geq a  \big\},\quad W^a_L:=\#\big\{ l\in\Lambda_L :\  |\Lambda_L||\varphi(l)|^2 < a  \big\}.\]
We have that their disjoint union is $\Lambda_L=S^a_L\cup W^a_L$ and there exists $\delta\in (0,1)$ such that $\sharp S^a_L=\delta\abs{\Lambda_L}$ and $\sharp W^a_L=(1-\delta)\abs{\Lambda_L}$. Otherwise we would have that either $\sharp S^a_L=\abs{\Lambda_L}$ or $\sharp S^a_L=0$, which together with the fact $\sum_{l\in\Lambda_L}\abs{\varphi(l)}^2=1$, leads to a contradiction.

By Lemma \ref{lem:upp-bound-phi}, we have that $\abs{\varphi(l)}^2\leq 2N/\abs{\Lambda_L}$, which implies the following
\begin{align}
1=\sum_{l\in\Lambda_L}\abs{\varphi(l)}^2 &\leq \sum_{l\in S^a_L}\abs{\varphi(l)}^2+ \sum_{l\in W^a_L}\abs{\varphi(l)}^2 \\
& \leq \frac{2N}{\abs{\Lambda_L}}\delta\abs{\Lambda_L} + \frac{a}{\abs{\Lambda_L}} (1-\delta)\abs{\Lambda_L}\\
& = 2N\delta + (1-\delta)a
\end{align}
This yields
\[ \delta \geq \frac{1-a}{2N-a}:=C_{a,N}.\]
Now inequality \eqref{lm_lower_bound} follows directly form \eqref{lm_lower_bound0}.
\end{proof}

\subsection{Upper bound}
Now we are ready to prove the upper bound on \eqref{eq:ub}.
\begin{lemma}\label{lem:upper-bound}
With the choice $L = \ceiling{\gamma E^{-\frac 1 {b}}}$ for some $\gamma>0$, we have the existence of a positive constant $C_2$ such that
\beq
\frac 1{|\Lambda|}\mathbb E\big(\Tr \1_{(-\infty,E]}\big((T^{\mathcal N,\mathcal N}_{g,\Lambda})^\beta + V_\omega \big)\big) \leq e^{-C_2E^{-\frac{1}{b}}}
\eeq
\end{lemma}
\begin{proof}
We need to find an upper bound for the r.h.s of \eqref{eq:bds-ub}.
We fix $a\in(0,1)$ depending only on $N$. Combining Theorem \ref{gen:temple} and Lemma \ref{lower_bound_lm2} we obtain
\begin{align}
\mathbb P \big(E_0(T_L^{\mathcal N}+ \wtilde V_\omega)<E\big)
&\leq
\mathbb P \left(\min_{\varphi\in \mathcal G_L,\|\varphi\|=1} \<\varphi,\wtilde V_\omega\varphi\> <
\frac 1 {(1-2\tilde c)^2}\Big( E + 6 \tilde c^2 \Big(\frac{C_0}{L^{b}}\Big)\Big) \right)\notag\\
&\leq
\mathbb P \left(\frac a {|\Lambda_L|} \sum_{k\in S^a_L} \wtilde V _\omega(k) <
\frac 1 {(1-2\tilde c)^2}\Big( E + 6 \tilde c^2 \Big(\frac{C_0}{L^{b}}\Big)\Big) \right)
\end{align}
with $\tilde c\in(0,1/2)$ as in Theorem \ref{gen:temple}. We fix $\tilde c = \frac 1 4$. Then we aim at estimating
\beq\label{eq:pot}
\mathbb P \left(\frac 1 {|\Lambda_L|} \sum_{k\in S^a_L} \wtilde V _\omega(k) <
\frac 4 a\Big( E +  \frac 3 8 \frac{C_0}{L^{b}}\Big) \right)
\eeq
By choosing
\beq
L = \ceiling{\gamma E^{-\frac 1 {b}}}
\eeq
for some $\gamma>0$ small, \eqref{eq:pot} is bounded above by

\beq\label{eq:ld}
\mathbb P \left(\frac{1}{\abs{\Lambda_L}} \sum_{k\in S_L^a} \wtilde V _\omega(k) <
\frac 4 a E\Big( 1+ \frac{3}{8}\frac{C_0}{\gamma^b}\Big) \right)
= \mathbb P \left(\frac {1}{\delta\abs{\Lambda_L}} \sum_{k\in S_L^a} \wtilde V _\omega(k) <
C_{\delta,a,\gamma,b,C_0} E \right)
\eeq
where $\delta\in(0,1)$, given in Lemma \ref{lower_bound_lm2}, such that $\delta \abs{\Lambda_L}=\sharp S_L^a$, and $C_{\delta,a,\gamma,b,C_0}= \frac{4}{a\delta } \left( 1+ \frac{3}{8}\frac{C_0}{\gamma^b} \right)$.
The term in the r.h.s of \eqref{eq:ld} can be estimated by a large deviation estimate, as in \cite[Lemma 6.4]{MR2509110}, which yields the existence of constants $C_1>$ and $C_2>$, depending on $a,N,\gamma,b,C_0$.  such that

\begin{align}
\mathbb P \left(\frac {1}{\delta \abs{\Lambda_L}} \sum_{k\in S_L^a} \wtilde V _\omega(k) <
C_{\delta,a,\gamma,b,C_0} E \right) & \leq e^{-C_1\delta \abs{\Lambda_L}}\\
& \leq e^{-C_2 E^{-\frac{1}{b}}}
\end{align}
where we used the fact that $\delta\geq C_{a,N}>0$. Plugging this in \eqref{eq:bds-ub} gives the claim.
\end{proof}

\section{Lower bound on the IDS}\label{s:low}
Let us recall that by Corollary \ref{corupplow}, for the IDS of the operator $T_f$ associated with a symbol $f$ satisfying Assumptions (A1)--(A4) in Section \ref{s:model}, there exists a constant $c>0$ such that for all $E\in\mathbb R$,
\beq N_{f_2}(E)\leq N_{f}(E)
\eeq
where $N_{f_2}$ is the IDS of the operator $T_{f_2}$ associated to a symbol $f_2$ given by
\beq
f_2(x)=C\big(2-2\cos(x-E_{i_0})\big)^{b/2}
\eeq
on $[-\pi,\pi)$, where $b= \max_{1\leq i \leq M} \beta_i$ and $i_0\in\{i:\beta_i=b \}$, see Assumption (A3).

\begin{lemma}\label{lem:equiv}
Let $f_2 = C \big(2-2\cos(x - E_{i_0})\big)^{b/2}$ and $f_3 = C\big(2-2\cos(x)\big)^{b/2}$, for $x\in[-\pi,\pi)$. Then for all $E\in\R$ the IDS of $f_2$ and $f_3$ satisfy
\beq
N_{f_2}(E) = N_{f_3}(E)
\eeq
\end{lemma}

\begin{proof}
We write $E=E_{i_0}$ and define the unitary $U:\ell^2(\Z)\to\ell^2(\Z)$
\beq
(Ux)(n) = e^{i E n } x(n).
\eeq
Then the symbol of $U^* T_{f_2} U$ is $f_2(\cdot + E) = f_3(\cdot)$ and since $V_\omega$ is diagonal
\beq
U^* H_{f_2} U = U^* (T_{f_2}+ V_\omega) U = T_{f_3}+ V_\omega.
\eeq
We denote by $1_{\Lambda_L}$ the characteristic function of $\Lambda_L=[-L,L]\cap\mathbb Z$, with $L\in\mathbb N$. Then for all $E\in\R$,
\begin{align}
N_f(E) = \lim_{L\to\infty} \frac{\Tr(1_{\Lambda_L} 1_{\leq E} (H_{f_2}) 1_{\Lambda_L})}{|\Lambda_L|}
=&
\lim_{L\to\infty} \frac{\Tr(U^*1_{\Lambda_L} U U^*1_{\leq E} (H_{f_2}) U U^*1_{\Lambda_L} U)}{|\Lambda_L|}\notag\\
=&
\lim_{L\to\infty} \frac{\Tr(1_{\Lambda_L} 1_{\leq E} (H_{f_3}) 1_{\Lambda_L} )}{|\Lambda_L|} = N_{f_3}(E),
\end{align}
where we used that $U^* 1_{\Lambda_L}  U = 1_{\Lambda_L} $, $ 1_{\leq E} ( U^*H_{f_3}U)=U^*1_{\leq E} (H_{f_3})U $ and the cyclicity of the trace.
\end{proof}

The above lemma implies that in order to bound from below the IDS of $T_{f_2}$, it is enough to bound from below the IDS of $T_{f_3}$, where $f_3$ is a function of type \eqref{eq:def-f} with $\alpha=b/2$ and $M=1$ there. We write

\[f_3=C\big(2-2\cos(x)\big)^{b/2}=C\big(\big(2-2\cos(x)\big)^{\tilde \alpha}\big)^{\tilde\beta}=C\tilde g(x)^{\tilde\beta}\]
with $\tilde \alpha:=\min\{n\in\N: b/2\leq n\}\in\N$, $\tilde\beta:= b/2\tilde\alpha\leq 1$ and
\[\tilde g(x)= \big(2-2\cos(x)\big)^{\tilde \alpha}.\]
Note that $\tilde g$ is the symbol of $(-\Delta)^{\tilde\alpha}$ with $\tilde \alpha\in\mathbb N$, which corresponds to a banded Laurent matrix with band width $2\tilde\alpha+1$.

Corollary \ref{upper_lower_bound} implies
\beq
N_{f_3}(E) \geq \frac 1{|\Lambda_L|}\mathbb E\big(\Tr \1_{(-\infty,E]}\big(C(T^{\mathcal D,\mathcal D}_{\tilde g,L})^{\tilde\beta} + V_\omega \big)\big)
\eeq
for all $L\geq 2\tilde \alpha+1$.
\begin{lemma}\label{lem:lowerbd}
We assume that the single-site probability distribution satisfies  $P_0([0,\epsilon))\geq C \epsilon^\kappa$ for some $C,\kappa>0$. By taking $L=\gamma'E^{-\frac{1}{b}}$ for some $\gamma>0$, there are positive constants $c_1,c_2$ such that
\beq \E\big( \Tr \1_{(-\infty,E]}\big(C(T^{\mathcal D,\mathcal D}_{\tilde g,L})^{\tilde\beta} + V_\omega \big)  \big) \geq
c_1E^{-\frac{1}{b}}e^{-c_2 \abs{\ln E} E^{-\frac{1}{b}}}.
\eeq
\end{lemma}
\begin{proof}
We estimate from below
\begin{align}
\E\big( \Tr \1_{(-\infty,E]}\big(C(T^{\mathcal D,\mathcal D}_{\tilde g,L})^{\tilde\beta} + V_\omega \big)  \big)
&\geq
\PP \big( E_0 \big(C(T^{\mathcal D,\mathcal D}_{\tilde g,L})^{\tilde\beta} + V_\omega \big) < E \big)\notag\\
&\geq
\PP \big( \big\<\psi, \big( C(T^{\mathcal D,\mathcal D}_{\tilde g,L})^{\tilde\beta} + V_\omega \big) \psi\big\> <  E \big)\notag \\
&\geq \PP \big( C\big\<\psi,  (T^{\mathcal D,\mathcal D}_{\tilde g,L})^{\tilde\beta}\psi\big\> + \max_{n\in\Lambda_L}V_\omega(n)  < E \big)
\end{align}
for any normalized $\psi\in\ell^2(\Lambda_L)$. Let $\psi_L$ be defined in Lemma \ref{lem:power} and $\phi_L:= \psi_L/\|\psi_L\|_2$. 
 With this choice of $\phi_L$  we further estimate using Jensen's inequality
\beq
\<\phi_L, (T^{\mathcal D,\mathcal D}_{\tilde g,L})^{\tilde\beta}  \phi_L\big\>
\leq
\<\phi_L, T^{\mathcal D,\mathcal D}_{\tilde g,L}  \phi_L\big\>^{\tilde\beta}.
\eeq
From the definition of $\psi_L$ it follows directly that
\beq
\|\psi_L\|_2^2 \geq C_1 L
\eeq
for some constant $C_1>0$ and Lemma \ref{lem:power} below gives for some other constant $C_2>0$
\beq
\<\psi_L, T^{\mathcal D,\mathcal D}_{\tilde g,L}  \psi_L\big\> \leq C_2 \frac 1 {L^{2 \tilde \alpha -1 }}.
\eeq
The last two inequalities impy the bound
\beq
C\<\phi_L, T^{\mathcal D,\mathcal D}_{\tilde g,L}  \phi_L\big\>^{\tilde\beta} \leq C_3L^{-2\tilde \alpha \tilde \beta} = C_3L^{-b},
\eeq
for some positive constant $C_3>0$. Altogether we obtain
\beq
E_0 \big(C(T^{\mathcal D,\mathcal D}_{\tilde g,L})^{\tilde\beta} + V_\omega \big)< \frac{C_3}{L^b}+\max_{n\in\Lambda_L}V_\omega(n) .  
\eeq
This last estimate replaces \cite[Eq.(4.51)]{AW15} and the rest follows along the lines of \cite[Sec. 4.4.2]{AW15}.
\end{proof}

\begin{lemma}\label{lem:power}
Let $N\in\mathbb N$ and  $\psi_1\in C^\infty([-1,1])$ satisfying $\psi_1(x) = 1$ for all $x\in (-\frac 1 2,\frac 1 2)$ and $\psi_1(x) = 0$ for $|x|\geq \frac 3 4$. We set $\psi_L\in C^\infty([-L,L])$, $\psi_L(x) := \psi_1(\frac x L)$, $x\in \Lambda_L$, and note that $\psi_L$ is supported on $\Lambda_{3L/4}$. There exists a constant $C_N$ depending on $N$ such that 
\beq
\<\psi_L, T^{\mathcal D,\mathcal D}_{ g,L}  \psi_L\big\> \leq C_{N} \frac 1 {L^{2N-1}}
\eeq
where $g(x) = (2 - 2 \cos(x))^N$, $x\in \mathbb T$. 
\end{lemma}

\begin{proof}
We note that $2 - 2 \cos(\cdot)$ is the symbol of the negative discrete Laplacian $-\Delta: \ell^2(\Z)\to\ell^2(\Z)$, $(-\Delta\psi)(n) = -\psi(n-1) + 2\psi(n) -\psi(n+1)$, $n\in\Z$. Since $\psi_L(x) = 0$ for all $|x|\geq \frac 3 {4L}$ and the boundary condition of $T^{\mathcal D,\mathcal D}_{ g,L}$ is supported in a neighborhood of the boundary of size $N$ only, we obtain for $L$ big enough
\beq\label{eq:T}
\<\psi_L, T^{\mathcal D,\mathcal D}_{ g,L}  \psi_L\big\> = \<\psi_L, (-\Delta)^N\psi_L\big\>.
\eeq
It is straight forward to see that $(-\Delta) = T^* T$ with $(T\psi)(n) = \psi(n) - \psi(n+1)$, $n\in\Z$ and therefore
\beq
\<\psi_L, (-\Delta)^N\psi_L\big\> = \<T^N\psi_L, T^N\psi_L\big\>.
\eeq
Assume we have for all $m\in\Lambda_L$ the representation 
\beq\label{eq:formula_diff}
\big|(T^N\psi_L)(m)\big| = \Big| \int_{m/L}^{(m+1)/L} \d t_1 \int_{t_1}^{t_1 + 1/L} \d t_2\, \cdots \int_{t_{N-1}}^{t_{N-1} + 1/L} \d t_N\, \psi^{(N)}_1(t_N)\Big|,
\eeq
where $\psi_1^{(N)}$ stands for the $N$th derivative of $\psi_1$. Using the latter formula, we readily obtain that
\beq
\<T^N\psi_L, T^N\psi_L\big\> \leq \frac{\|\psi_1^{(N)}\|^2_\infty}{L^{2N}} \sum_{m\in\Lambda_L} 1 <  \frac{3\|\psi_1^{(N)}\|^2_\infty}{L^{2N-1}}
\eeq
which gives the assertion together with \eqref{eq:T}. 

It remains to show formula \eqref{eq:formula_diff}. This follows from an induction with respect to $N$ and the fundamental theorem of calculus.
\end{proof}

\section{Proof of Theorem \ref{thm:main}}\label{s:proof}
\begin{proof}We combine the results from the two previous sections. The upper bound is given by combining Corollary \ref{corupplow}, Corollary \ref{upper_lower_bound}, the discussion preceding Equation \eqref{eq:ub} and Lemma \ref{lem:upper-bound}.

Under the supplementary assumption on the single-site probability distribution, the lower bound is obtained by combining Corollary \ref{corupplow}, Corollary \ref{upper_lower_bound}, Lemma \ref{lem:equiv} and the discussion preceding it, and Lemma \ref{lem:lowerbd}.
\end{proof}

\appendix
\section{}
\subsection{Off-diagonal decay of Laurent matrices}
\begin{lemma}\label{lem:offdiag-decay}
Let $T_f$ be Laurent operator associated to a symbol $f$ satisfying Assumptions (A1)\,--\,(A2). The  matrix entries of $T_f$ decay as
\[\abs{a_{n-m}}\leq \frac{1}{\abs{n - m}^{1+\nu}}, \quad m, n \in\mathbb Z\]
therefore, $T_f$ is an, at most, long-range operator with polynomially decaying off-diagonal terms.
\end{lemma}
\begin{proof}
To see this, note that $f\in C^{1}(\mathbb T)$ except at finitely many points $x_j\in \mathbb T$, $j=1,..,m$ for some $m\in \mathbb N$ where $f$ remains $\nu$-H\"older continuous for some $\nu>0$ at these points.
An adaption of the proof of \cite[Thm. 2.2]{GRM20} implies for symbols $h\in C^1(\mathbb T)\setminus\{y\}\cap C^{\alpha}(\mathbb T)$ that
\beq
\label{appendix:off_diag_decay} |T_h(n,m)|\leq \frac C {|n-m|^{1+\nu}}, \quad n,m\in\Z
\eeq
for some $C>0$.
We write $g = \sum_{j=1}^m h_j$ with $h_j \in C^1(\mathbb T)\setminus\{x_j\} \cap C^\alpha(\mathbb T)$. Now we apply inequality \eqref{appendix:off_diag_decay} to each function $h_j$ and since there is only a finite number of those we end up with the result.
\end{proof}

\begin{proposition}\label{prop:IDS-Tf}
Let $T_f$ be the Laurent operator associated to a symbol $f$ satisfying Assumptions (A1)\,--\,(A2). The limit
\beq
\wtilde I_f(E):=\lim_{L\to\infty} \frac{\Tr\big( 1_{[-L,L]} 1_{\leq E} (T_{f}) \big)}{2L+1}
\eeq
exists and equals the IDS of $T_f$, $I_f(E)$, for all $E\in\mathbb R$, where $I_f$ is defined in \eqref{eq:IDS-Tf}. Moreover,
\beq\label{eq:IDS-Tf-form}
I_f(E)= \frac 1 {2\pi}\big| \big\{ k\in [-\pi,\pi]: \ f(k) \leq E\big\}\big|,
\eeq
\end{proposition}
\begin{proof}
Using the off-diagonal decay of $T_f$ shown above we can follow the arguments in \cite[Proposition 2.1]{GRM20} to show that $\wtilde I_f$ and $I_f$ are identical for all $E\in\mathbb R$.

Next, let $\tau_x:\ell^2(\mathbb Z)\rightarrow \ell^2(\mathbb Z)$ be the translation by $x\in\mathbb Z$ acting on  $\ell^2(\mathbb Z)$ by $\tau_x \varphi(n)=\varphi(n-x)$, for $\varphi\in\ell^2(\mathbb Z)$.
The translation invariance of $T_f$ implies that
\beq \lim_{L\to\infty} \frac{\Tr\big( 1_{[-L,L]} 1_{\leq E} (T_{f}) \big)}{2L+1}=\frac{1}{2L+1}\sum_{x\in [-L,L]}\langle{\delta_x,1_{\leq E} (T_{f})\delta_x}\rangle= \langle \delta_0, 1_{\leq E} (T_{f})\delta_0 \rangle
\eeq
Using Fourier transform in the r.h.s. of the last line and the fact that $T_f$ is unitary equivalent to the operator $M_f$, multiplication by $f$, yields \eqref{eq:IDS-Tf-form}.
\end{proof}

\subsection{Properties of the ground state space $\mathcal G$ of Neumann restrictions of Toeplitz matrices}
Let $T$ be a Laurent matrix with band width $N$ associated to a symbol $g$ of the form \eqref{eq:def-g}. Consider its restriction to the cube $\L_L=[-L,L]\cap \mathbb Z$, with $L\in\mathbb N$, $L>2N+1$ with modified Neumann boundary conditions (see Section 3), denoted by $T_L^{\mathcal N}$. We define its ground state space by
\beq
\mathcal G:=\big\{\varphi\in \ell^2(\Lambda_L): T_L^{\mathcal N}\varphi = 0\big\}
\eeq
We recall from \cite[Section 5]{G20} (see section 3) that the ground-state space $\mathcal G$ is spanned by
\beq
\mathcal G = \text{span} \big\{\varphi^j_{k,L}\in \ell^2(\Lambda_L):\ k=1,...,M,\,\,j=0,...,\overline \alpha-1  \big\}
\eeq
with
\begin{align}\label{def:varphi_k}
\varphi^j_{k,L} := \frac 1 {K_{j,L}^{1/2}}\big( (-L)^j e^{i L  E_k },...,0^j,1^j e^{-i E_k},..., L^je^{-i L E_k}\big)^T\in\R^{2L+1}
\end{align}
 for $k=1,...,M$ and $j=0,...,\overline \alpha-1$, and
 \beq
 K_{j,L} = 2\sum_{m=-L}^L |m|^{2j}.
 \eeq
\begin{lemma}
Let $\varphi\in\mathcal G$. Then there exists a constant $C>0$ such that for all $L\in\N$, $r\neq k$ and $j,s\in \{0,...,\overline \alpha\}$
\beq
|\<\varphi_{k,L}^j,\varphi_{r,L}^s\>| \leq \frac C {|\Lambda_L|}.
\eeq
\end{lemma}

\begin{proof}
We compute
\beq
|\<\varphi_{k,L}^j,\varphi_{r,L}^s\>| = \frac 1 {(K_{j,L}K_{s,L})^{1/2}}\Big|\sum_{m=-L}^L m^{j+s} e^{i m (E_k - E_r)}\Big|
\eeq
Using summation by parts $j+s$ times and $|\sum_{m=-L}^L e^{i m (E_k - E_r)}| = O(1)$ we obtain that
\beq
\Big|\sum_{m=-L}^L m^{j+s} e^{i m (E_k - E_r)}\Big| = O(L^{j+s}).
\eeq
Now $K_{j,L}K_{s,L} = O(L^{2j+2s+2})$ implies
\beq
|\<\varphi_{k,L}^j,\varphi_{r,L}^s\>| = O(L^{-1})
\eeq
and therefore there exists a constant $C_{k,j,r,s}$ depending on $k,j,r,s$ but independent of $L$ such that
\beq
|\<\varphi_{k,L}^j,\varphi_{r,L}^s\>| \leq \frac{C_{k,j,r,s}}{L}.
\eeq
Taking $C:=\displaystyle \max_{j,s,r\neq k} C_{k,j,r,s}$, gives the assertion.
\end{proof}

\begin{lemma}\label{lem:upp-bound-phi}
There exists $L_0$ such that for all $L\geq L_0$, $l\in \Lambda_L$ and $\varphi\in\mathcal G$
\beq
|\varphi(l)| \leq \frac{\sqrt{2N}}{\sqrt{|\Lambda_L|}}.
\eeq
\end{lemma}

\begin{proof}
Let $\varphi\in \mathcal G $ with $\|\varphi\|=1$. Then
$
\varphi =\displaystyle \sum_{k=0}^{n-1}\sum_{j=0}^{\overline \alpha - 1} a_{j,k} \varphi_{k,L}^j
$ for some $a_{j,k}\in\C$
and $\varphi_{k,L}^j$ given in \eqref{def:varphi_k}. We compute
\beq
\<\varphi,\wtilde V_\omega\varphi\>  = \sum_{l\in \Lambda_L} \tilde V_\omega(l) |\varphi(l)|^2
\eeq
where $\varphi(l)$ stands for the $l$th component of the vector $\varphi$. The last lemma implies that $|\<\varphi_{k,L}^j,\varphi_{r,L}^s\>| \leq \frac C {|\Lambda_L|}$ for some constant $C>0$ and all $k\neq r$. Hence, we obtain
\begin{align}
1 = \sum_{l\in \Lambda_L}|\varphi(l)|^2
& =  \sum_{k,r=0}^{n-1}\sum_{j,s=0}^{\overline\alpha - 1} \overline{a}_{k,j} a_{r,s} \<\varphi_{k,L}^j,\varphi_{r,L}^s\> \notag\\
& \geq  \sum_{k=0}^{n-1}\sum_{j=0}^{\overline\alpha - 1} |a_{k,j}|^2 - \sum_{(k,j)\neq(r,s)} |a_{k,j}||a_{r,s}| \frac C {|\Lambda_L|} \notag\\
&\geq  \sum_{k=0}^{n-1}\sum_{j=0}^{\overline\alpha - 1} |a_{k,j}|^2  - \frac 1 2\sum_{(k,j)\neq(r,s) } \big(|a_{k,j}|^2 + |a_{r,s}|^2 \big) \frac C {|\Lambda_L|} \notag\\
&\geq  \sum_{k=0}^{n-1}\sum_{j=0}^{\overline\alpha - 1} |a_{k,j}|^2 \Big( 1-  \frac {C N} {|\Lambda_L|}\Big)
\end{align}
where we used the inequality $|xy|\leq \frac 1 2(|x|^2 + |y|^2)$ for $x,y\in\R$.
We choose $L_0\in\N$ such that for all $L\geq L_0$ we obtain $\frac {C N} {|\Lambda_L|}\leq \frac 1 2$ and therefore
\beq
 \sum_{k=0}^{n-1}\sum_{j=0}^{\overline\alpha - 1} |a_{k,j}|^2 \leq 2.
\eeq
This implies
\beq\label{bound_varphi(l)}
|\varphi(l)|
\leq \frac 1 {\sqrt {|\Lambda_L|}}  \sum_{k=0}^{n-1}\sum_{j=0}^{\overline\alpha - 1} |a_{k,j}|
\leq  \frac {\sqrt N} {\sqrt {|\Lambda_L|}} \big(\sum_{k=0}^{n-1}\sum_{j=0}^{\overline\alpha - 1} |a_{k,j}|\big)^{1/2}
\leq \frac {\sqrt{2N}} {\sqrt {|\Lambda_L|}}.
\eeq
\end{proof}

\section*{Acknowledgements}
CRM acknowledges the Agence Nationale de la Recherche for their financial support via ANR grant RAW ANR-20-CE40-0012-01.
MG thanks Peter M\"uller and Jacob Shapiro.

\newcommand{\etalchar}[1]{$^{#1}$}

\end{document}